\documentclass[a4paper,UKenglish]{lipics-v2018}

\usepackage{amsmath,amssymb,amsthm,graphicx}
\usepackage{mathtools, arcs}
\usepackage[dvipsnames]{xcolor}
\usepackage{soul}
\nolinenumbers
\usepackage{microtype}%if unwanted, comment out or use option "draft"

\author{Farnaz Sheikhi}{School of Computer Science, Institute for Research in Fundamental Sciences (IPM), Tehran, Iran.}{f.sheikhi@ipm.ir}{}{}%mandatory, please use full name; only 1 author per \author macro; first two parameters are mandatory, other parameters can be empty.

\author{Sharareh Alipour}{School of Computer Science, Institute for Research in Fundamental Sciences (IPM), Tehran, Iran.}{alipour@ipm.ir}{}{}

\authorrunning{F. Sheikhi and S. Alipour}%mandatory. First: Use abbreviated first/middle names. Second (only in severe cases): Use first author plus 'et al.'

\Copyright{F. Sheikhi and S. Alipour}%mandatory, please use full first names. LIPIcs license is "CC-BY";  http://creativecommons.org/licenses/by/3.0/

\subjclass{
}% mandatory: Please choose ACM 2012 classifications from https://www.acm.org/publications/class-2012 or https://dl.acm.org/ccs/ccs_flat.cfm . E.g., cite as "General and reference $\rightarrow$ General literature" or \ccsdesc[100]{General and reference~General literature}.

%mandatory

\EventEditors{John Q. Open and Joan R. Access}
\EventNoEds{2}
\EventYear{2018}
\EventDate{December 16-19, 2018}
\EventLocation{Jiaoxi, Yilan County, Taiwan}
\EventLogo{}
\SeriesVolume{42}
\ArticleNo{23}
  \newif\ifComments
\Commentstrue
\ifComments
  \newcommand{\farnaz}[1]{\textcolor{blue}{Farnaz: #1}}
  \newcommand{\sharare}[1]{\textcolor{LimeGreen}{Sharare: #1}}

\else
  \newcommand{\farnaz}[1]{}
  \newcommand{\sharare}[1]{}
  \newcommand{\rmv}[1]{}
\fi

\newif\ifOmit
\Omittrue

\theoremstyle{definition}
\newtheorem{observation}[theorem]{Observation}

\begin{document}
\title{On Triangluar Separation of Bichromatic Point Sets in Polygonal Environment}
\maketitle

\begin{abstract}
Let $\mathcal P$ be a simple polygonal environment with $k$ vertices in the plane. Assume that a set $B$ of $b$ blue points and a set $R$ of $r$ red points are distributed in $\mathcal P$. We study the problem of computing triangles that separate the sets $B$ and $R$, and fall in  $\mathcal P$. We call these triangles \emph{inscribed triangular separators}. We propose an output-sensitive algorithm to solve this problem in $O(r \cdot (r+c_B+k)+h_\triangle)$ time, where $c_B$ is the size of convex hull of $B$, and $h_\triangle$ is the number of inscribed triangular separators. We also study the case where there does not exist any inscribed triangular separators. This may happen due to the tight distribution of red points around convex hull of $B$ while no red points are inside this hull. In this case we focus to compute a triangle that separates most of the blue points from the red points. We refer to these triangles as \emph{maximum triangular separators}. Assuming $n=r+b$, we design a constant-factor approximation algorithm to compute such a separator in $O(n^{4/3} \log^3 n)$ time. 
"Eligible for best student paper"
\end{abstract}

\section{Introduction}

A vast area of research in computational geometry focuses on optimization problems on point sets. An interesting category of them arises from facility location applications where customers are modeled by points in the plane, and the area of service is modeled by a geometric shape covering the points. This motivates the interesting class of \emph{covering problems} where the goal is to cover all points by a geometric shape (the \emph{cover}) with the minimum area/perimeter. Techniques and time complexity of the algorithms proposed vary according to the properties of the shape of cover. This problem has been studied for circular~\cite{Megiddo-LP}, rectangular~\cite{rotating-cal}, L-shaped~\cite{L-cover}, and arbitrary triangular~\cite{tri-rourke,tri-chang,tri-ag} covers. The problem has also attracted attentions for covers in less general forms:  isosceles triangular cover~\cite{bose1} and triangular cover with a fixed angle~\cite{bose2} for instances.

Gradually researches focus on inducing more realistic assumptions on points and the environment of covering. Naturally, there can be desirable and undesirable points to service, where covering the undesirable points together with the desired ones may not only be useful but also be harmful. This discrimination in data is often modeled by assigning different colors to points. There it is assumed that a set $B$ of $b$ blue points (modeling the desirable customers) and a set $R$ of $r$ red points (modeling the undesirable customers) are given.  Then covering problem can be viewed as a \emph{geometric separability problem} where the goal is to use a geometric shape (the \emph{separator}) to cover $B$ without covering any points in $R$. The separability problem has been studied for linear~\cite{Megiddo-LP}, circular~\cite{sep-cir}, rectangular~\cite{sep-rect, kreveld-tech, farnaz-rect}, and L-shaped~\cite{sep-ours} separators.
For the rectangular separator, motivated by the reconstruction of urban scenes from LIDAR data, the separability problem has also been studied to find the rectangular separators that fall in a simple closed curvature defined by a collection of circular arcs~\cite{sep-rect, kreveld-tech}.

The separability problem is also studied for more general class of separators: convex~\cite{conv-sep} and simple~\cite{simple-fekete, simple-apprx-mitchel} polygonal separator with the minimum number of edges. The problem of finding a convex polygon with the minimum number of edges separating the two point sets, if it exists, can be solved in $O((r+b) \log(r+b))$ time~\cite{conv-sep}. However,  omitting the convexity condition of the separator in the above problem  makes it NP-complete~\cite{simple-fekete, simple-apprx-mitchel}.

The majority of separability problems are studied under the assumption that the sets $R$ and $B$ are freely distributed in the plane. However, in the real world problems, geographical constraints impose restrictions on the feasible locations of the separator. This motivates the study of separability problem for the case where the  point sets $B$ and $R$ are distributed in a polygonal environment $\mathcal P$, and the goal is to find separators that separate $B$ from $R$ and fall in $\mathcal P$~\cite{demaine-chord}. It has been shown that the problem of finding the minimum number of non-crossing \emph{chords}\footnote{A chord is defined as a segment that connects two points on the boundary of $\mathcal P$ and falls inside $\mathcal P$.} separating the two point sets in $\mathcal P$ is polynomially solvable when $\mathcal P$ is very simple-- namely, a strip or a triangle. However,  this problem becomes NP-complete when $\mathcal P$ may have holes. In between, there is an intriguing open problem~\cite{demaine-chord}.

In this paper we study the separability of $B$ and $R$, where the sets are distributed in a a polygonal environment $\mathcal P$. We use arbitrary triangles as the separators. (Our triangular separators neither need to have fixed corner angles nor the same side length.) We propose an output-sensitive algorithm running in $O(r \cdot (r+c_B+k)+h_\triangle)$  time to find $h_\triangle$ triangular separators that separate $B$ from $R$ and falls in the polygonal environment $\mathcal P$. We call these separators \emph{inscribed triangular separators}. Our algorithm can also handle the case where $\mathcal P$ is a simple closed curvature defined by a collection of circular arcs simulating the density of points as in~\cite{sep-rect,kreveld-tech}. We design the lower bound of $\Omega(r^3)$ on the number of triangular separators. It may happen the case where there does not exist any inscribed triangular separators. This may happen due to the tight distribution of red points around convex hull of $B$ while no red points are inside this hull. In this case we design a constant-factor approximation algorithm that finds a triangle that separates most of the blue points from the red points (referred as the \emph{maximum triangular separator}). Assuming $n=r+b$, the proposed approximation algorithm runs in $O(n^{4/3} \log^3 n)$ time.

\section{Preliminaries}
Let $B$ be a set of $b$ blue points and $R$ be a set of $r$ red points. We assume that points are in general position, that is no three points are collinear and no three lines have a common intersection. For more clear explanation of our technique, for now we assume that points are located in the plane, and ignore the restriction that $B$ and $R$ are inside a polygonal environment $\mathcal P$. We will get back to it later.

Let $CH_B$ denote the convex hull of blue points with $c_B$ vertices.
It is easy to see that every triangular separator has to enclose $CH_B$. Hence, to have such a separator, a necessary condition is that $CH_B$ does not contain any red points. To find the triangular separators we  make them well-defined. Note that by translation and rotation, we can change any triangular separator into a separating triangle with each side (or the line through each side) passing through a red and a blue point. We call these modified triangular separators \emph{canonical triangular separators}, and in the rest of the paper we focus to find these separators. It is easy to see that there are $\Omega(r^3)$ canonical triangular separators in the worst case. See Figure~\ref{lower-bn}. From each red point we draw tangents to $CH_B$. This results in the arrangement $\mathcal A$ with $2r$ lines in the plane.

\begin{observation}
The vertices of canonical triangular separators lie on the vertices of arrangement $\mathcal A$.
\end{observation}

\begin{figure}[h]
\centering
\includegraphics[width=5cm]{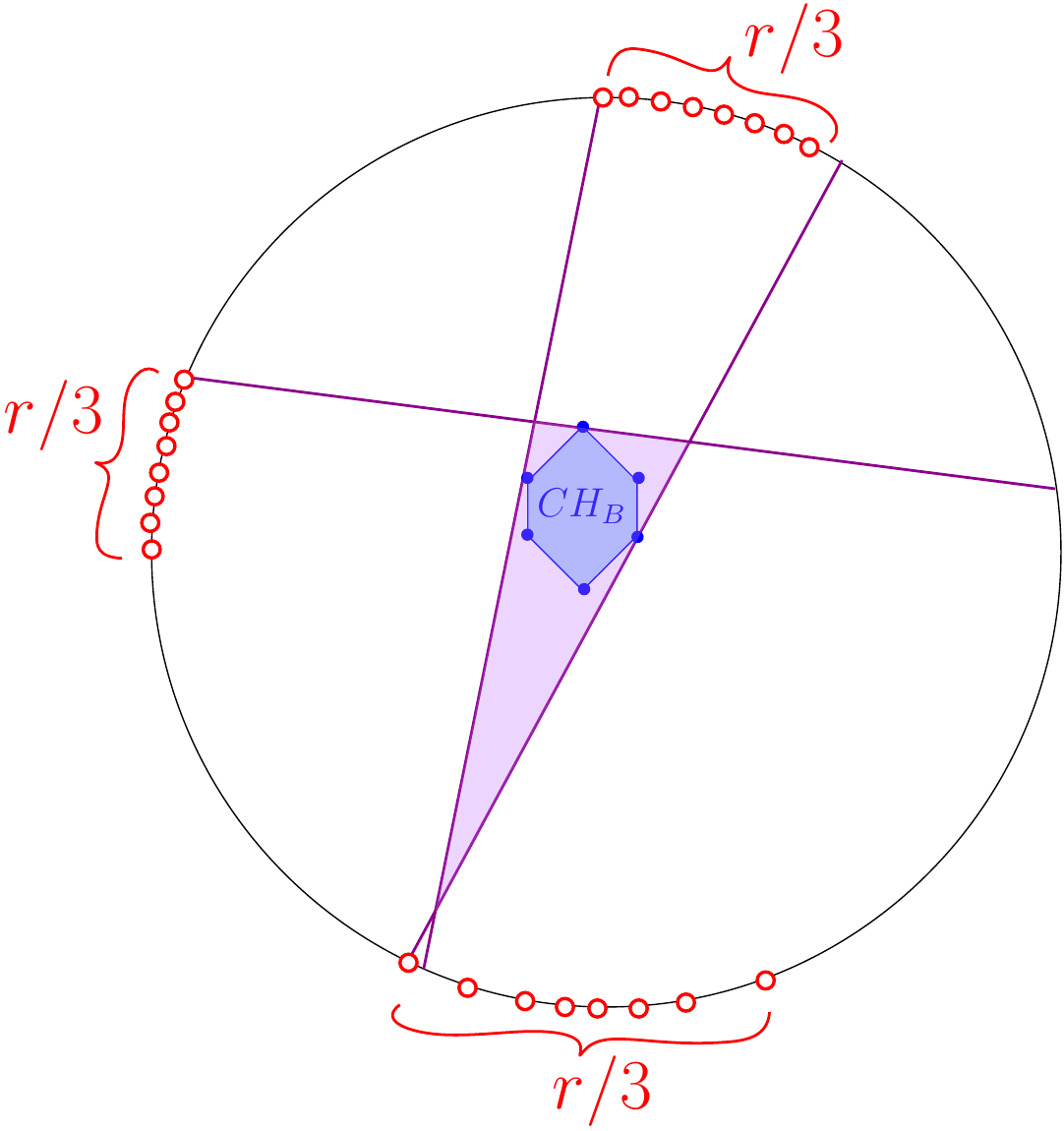}
\caption{Lower bound on the number of canonical triangular separators. }
\label{lower-bn}
\end{figure}

Arrangement $\mathcal A$ has $O(r^2)$ vertices (the candidate locations for the vertices of canonical triangular separators). In Section~\ref{sec-rank}, we propose a ranking strategy that helps  to process each vertex of $\mathcal A$ in $O(1)$ time. To get to that we need to define some concepts.

\begin{definition}
For a point $p$ we define the \emph{semi-triangle based on $p$}, denoted by $\triangle_p$, as follows. Consider tangents from $p$ to $CH_B$. Let $b_{\mathrm{left}}$ and $b_{\mathrm{right}}$ be the tangent points on $CH_B$. The region bounded by the directed segments $\overrightarrow{p b_{\mathrm{left}}}$ and $\overleftarrow{p b_{\mathrm{right}}}$ and the clockwise chain from $b_{\mathrm{left}}$ to $b_{\mathrm{right}}$ on $CH_B$ defines the \emph{semi-triangle based on $p$} or $\triangle_p$ for short.
\end{definition}

For a point $p$ in the plane, $\triangle_p$ is illustrated in Figure~\ref{semi-tri-p}. By $\triangle_p$ we consider the region excluding the boundary.
\begin{figure}[h]
\centering
\includegraphics[width=4cm]{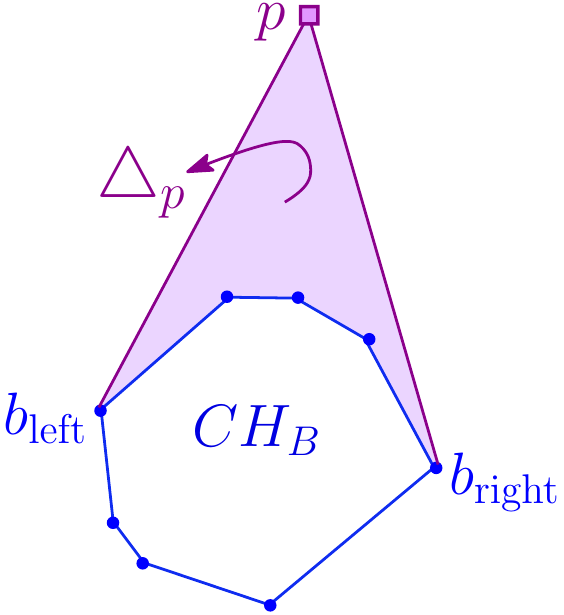}
\caption{For a point $p$, the shaded area shows  the semi-triangle based on $p$ ($\triangle_p$). }
\label{semi-tri-p}
\end{figure}

\section{Ranking strategy}~\label{sec-rank}
It is time to look at the separability problem in our polygonal environment $\mathcal{P}$.
Arrangement $\mathcal{A}$ has $O(r^2)$ vertices in the plane. Now let's look at the part of $\mathcal{A}$ that falls in the polygonal environment $\mathcal{P}$. Clearly each line in $\mathcal{A}$ can intersect $\mathcal{P}$ several times (according to the shape of $\mathcal{P}$). However, since we aim to find the inscribed triangular separators, it is only the first intersection that matters (by the first  we mean the first intersection  from  a vertex in $\mathcal{A}$ in outward direction from $CH_B$). Hence, we can ignore the rest of intersections. This way $\mathcal{A}$ has at most $2r$ vertices on the boundary of $\mathcal P$.  We omit the part of $\mathcal{A}$ that falls outside $\mathcal P$, and from now on by $\mathcal{A}$ we mean the part of this arrangement that falls inside $\mathcal P$.  This way $\mathcal{A}$ can be treated as a geometric graph. We can classify vertices of $\mathcal{A}$ into four following types:
\begin{itemize}
\item \emph{Type-I}: Vertices that coincide with red points.
\item \emph{Type-II}: Vertices that coincide with blue points.
\item \emph{Type-III}: Vertices that lie on the boundary of $\mathcal P$.
\item \emph{Type-IV}: Vertices that fall in non of the above groups.
\end{itemize}
In this section we design a \emph{ranking strategy} for vertices of $\mathcal{A}$. Having this strategy in the pre-processing stage, in the main algorithm for an arbitrary vertex $v\in \mathcal{A}$, we can decide in $O(1)$ time whether $\triangle_v$ contains a red point or not. Ranking strategy consists of two steps: an \emph{initialization step} and an \emph{iterative step}.

In the initialization step we find the set of \emph{ancestors}. For each vertex $b$  of $CH_B$ which is a tangent point, we define two ancestors as follows. If there is just a single line in $\mathcal{A}$ passing through $b$, then ancestors of $b$ are the first vertices that we visit once we walk on the line from $b$ upwards and downwards. Otherwise, if there are several lines in $\mathcal{A}$ passing through $b$, then we take the line with the minimum slope, and consider the first vertices that we visit once we walk on this line from $b$ upwards and downwards.
Considering all tangent point vertices of $CH_B$, let $\phi$ denote the set of ancestors. It is observed that

\begin{observation}~\label{empty-tri}
For every vertex $v \in \phi$, $\triangle_v$ contains no red points.
\end{observation}

In Figure~\ref{propagation-pic}, we have $\phi =\{v_{22}, v_{23}, v_{26}, v_{29}, v_{31}, v_{33}, v_{35}\}$. We set two labels called \emph{rank} and \emph{level} for the vertices in $\phi$.
According to Observation~\ref{empty-tri}, for every vertex $v \in \phi$ we set $rank(v)=0$. Further, since these vertices are processed in the first step of our strategy we set the level of all of them by $0$ as well. See~Table~\ref{our-table}.

Now that the initialization step is over, the iterative step of the strategy is to define the rank of neighbors of vertices in $\phi$ in the geometric graph $\mathcal{A}$. (Note that the vertices $u$ and $v$ are neighbors in graph $\mathcal A$ if they share an edge.) The overall idea is to propagate the rank of ancestors through vertices of $\mathcal A$ if these vertices are of types other than Type-I vertices (vertices that coincide with red points), and increase the propagated rank if we reach a Type-I vertex. The details are as follows.
For a vertex $v \in \phi$, let $\mathcal{N}(v)$ denote the set of  neighbors of $v$. Next we explain how to specify the rank and the level of these vertices.
For a vertex $u \in \mathcal{N}(v)$, we set $level(u)=1$. However, to specify $rank(u)$, type of $u$ should be considered as well:
If $u$ is a Type-I vertex (it coincides with a red point), then $rank(u)=1$. Otherwise, $rank(u)=0$.

Once the vertices with level $i$ is processed, the procedure continues with processing the vertices with level $i+1$. Let $\phi_{i+1}$ be the set of vertices whose level is $i+1$. Then for a vertex $v \in \phi_{i+1}$, we set the rank and the level of neighbors of $v$ in $\mathcal A$ as follows. Let $u$ be a neighbor of $v$ in $\mathcal A$. Then $level(u)=level(v)+1$, resulting in $level(u)=i+2$. To specify $rank(u)$, if $u$ is a Type-I vertex then $rank(u)=rank(v)+1$. Otherwise, $rank(u)=rank(v)$.
Note that deviation in the direction of propagating the rank of a Type-I vertex also increases the rank as follows. Throughout the ranking strategy, once we visit a Type-I vertex $x$ lying on the intersection of lines $\ell$ and $\ell'$ (for instances) we tag it as the red parent of $\ell$ and $\ell'$. Now let $y$ be a vertex on $\ell$, visited later than $x$ in the strategy. And assume $z_1$ and $z_2$ are the neighbors of $y$ to rank, where $z_2$ lies on $\ell$ while $z_1$ does not. See Figure~\ref{rank-det}. Then to rank $z_1$, since this vertex has a deviation from the red parent (it leaves the line $\ell$),  independent of type of the vertex $z_1$, there is an increase in rank. That is, $rank(z_1)=rank(y)+1$ due to the deviation. If $z_1$ is also a Type-I vertex then the rank  gets increased again by 1 due to the property.
Note that if the vertex $u$ has been ranked before, then $rank(u)$ is updated with the maximum of the old and the new rank. The level is also updated if the rank gets updated. When we reach a Type-III vertex (the vertex that lie on the boundary of $\mathcal P$)  we do not process their neighbors since the unprocessed neighbors lie on the boundary of $\mathcal P$, and there is no need to proceed. The procedure is finished when we have no more vertices to proceed.

\begin{figure}[h]
\centering
\includegraphics[width=10cm]{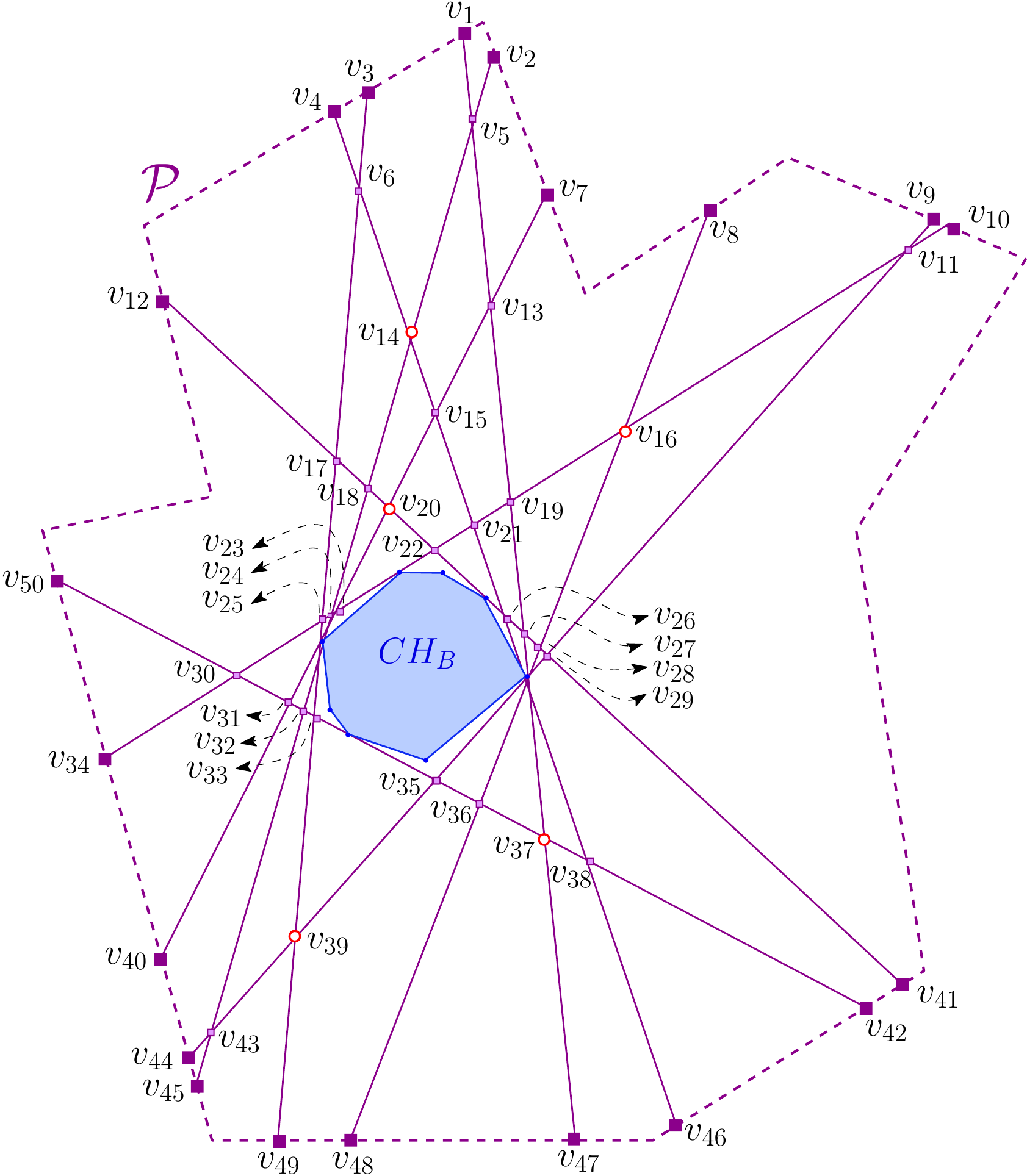}
\caption{Following the ranking strategy on arrangement $\mathcal A$ embedded in the polygonal environment $\mathcal P$. }
\label{propagation-pic}
\end{figure}

\begin{figure}[h]
\centering
\includegraphics[width=5cm]{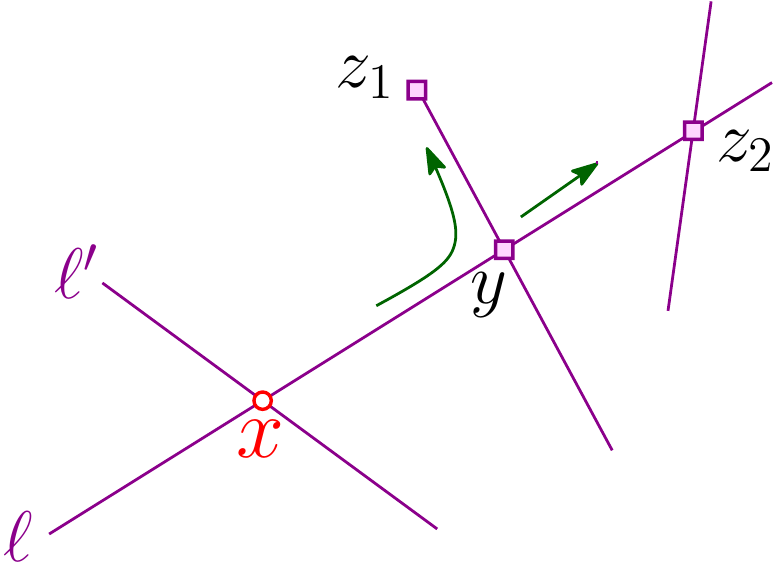}
\caption{Illustrating how deviation can increase the rank.}
\label{rank-det}
\end{figure}

\begin{lemma}
In an inscribed triangular separator with the vertex $v \in \mathcal A$, we have  $rank(v) <2$.
\end{lemma}

\begin{proof}
The proof is by contradiction. Let $v \in \mathcal{A}$ be a vertex of an inscribed triangular separator, and assume by contradiction that $rank(v) \geq 2$. Consider the tangent points of the lines trough $v$ on $CH_B$. Let $b_{\mathrm{left}}$ and $b_{\mathrm{right}}$ be these blue tangent points. Let $a$ be an ancestor in $\triangle_v$. Since $rank(v) \geq 2$ then $a \neq v$. Further, since points are in general position, then there exists at most one red point on the segment $\overline{vb_{\mathrm{left}}}$. The same is true for the segment $\overline{vb_{\mathrm{right}}}$. Thus, using the path $\overline{vb_{\mathrm{left}}}$ or $\overline{vb_{\mathrm{right}}}$, $rank(a)$ is at most increased by $1$ to get to $v$, resulting that $rank(v)$ gets at most $1$. Thus, there should exist a path from $a$ to $v$ that increases the propagated rank more. Therefore, this path should visit at least one red point. See Figure~\ref{proof-1}. However, in this case there should exist at least one red point in $\triangle_v$, contradicting that there is an inscribed triangular separator (having no red point inside) with a vertex on $v$.
\end{proof}
\begin{figure}[h]
\centering
\includegraphics[width=6cm]{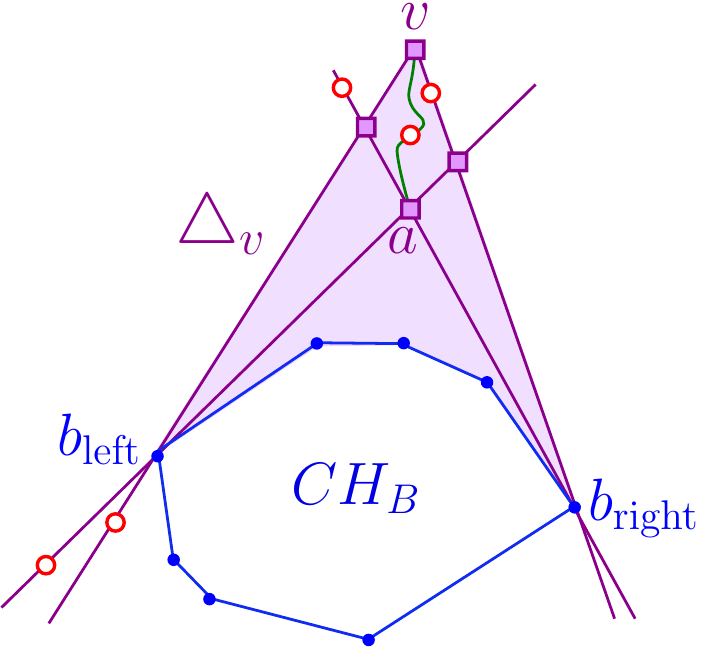}
\caption{Studying $rank(v)$.}
\label{proof-1}
\end{figure}

We also need to store another label called \emph{extreme} for vertices of Type-I, Type-II, and Type-IV (we do not store this label for the vertices on the boundary of $\mathcal P$ unless they are also a vertex in the arrangement $\mathcal A$ in the plane). Consider a vertex $v \in \mathcal A$ and the directed half-line from $v$ and tangent to $CH_B$. Then walking in the direction of this half-line, we define $extreme(v)$ as the last vertex whose rank is less than $2$ and is seen in the direction of walking. See Table~\ref{our-table} for a complete calculation of the labels rank, level, and extreme. We will see how to use these pre-computed labels in the main algorithm. For now, let us examine the time complexity of the ranking strategy.

\begin{theorem}
The ranking strategy takes $O(r \cdot (r+c_B+k))$ time.
\end{theorem}

\begin{proof}
Embedding arrangement $\mathcal A$ in the polygonal environment $\mathcal P$ with $k$ vertices, and computing the first intersection for each line of the arrangement takes $O(r \cdot k)$ time.
Then, the ranking strategy contains two steps: the initialization step and the iterative one.
In the initialization step, for each vertex $b$ of $CH_B$ (which is a tangent point) we calculate  the two  corresponding ancestors. To do so we proceed as follows. If there is just a single line in $\mathcal A$ passing through $b$, then ancestors of $b$ in the directions upwards and downwards can be computed in $O(r)$ time. Otherwise, if there are several lines in $\mathcal A$ passing through $b$, then we take the line with the minimum slope, and compute the two corresponding ancestors of $b$ on this line. This also takes $O(r)$ time. Assuming $c_B$ to be the size of $CH_B$, then finding the set of ancestors takes  $O(r \cdot c_B)$ time.

Having completed the initialization step, in the iterative step we compute the rank and level of neighbors of the vertices that are ranked just a step before. The ranking strategy specifies the rank and the level of each vertex in $O(1)$ time once the vertex takes its turn. Thus, the iterative step takes $\sum_{i=1}^{4r^2+2r} \deg(v_i)$, where $v_i$ is a vertex of the geometric graph $\mathcal A$, and $\deg(v_i)$ denotes its degree. This summation equals 2 times the number of edges that is $O(r^2)$.

We also compute the extreme label for vertices in $\mathcal A$. To do so
for each line $\ell$ in $\mathcal A$, we find the top-most vertex in $\mathcal A$ in $O(r)$ time. Once  we reach the top-most vertex on $\ell$ (let it be $v_\ell$) we can compute $extreme(v_\ell)$ in $O(r)$ time, and use its value for all vertices on $\ell$ from $v_\ell$ till $extreme(v_\ell)$. Thus, computing the extreme labels takes $O(r)$ time per line, leading to the total time of $O(r^2)$.

Therefore, the ranking strategy takes $O(r \cdot (r+c_B+k))$ time in total.
\end{proof}

\begin{table}[h!]
  \centering
  \begin{tabular}{  l | l | l | l| l| l| l| l| l|l|l }
    %\hline
     & $v_1$ & $v_2$ & $v_3$ & $v_4$ & $v_5$ & $v_6$ & $v_7$ & $v_8$ & $v_9$ & $v_{10}$\\
     \hline
    Rank &  3  &  3 &  3 &  3 & 3  & 3  & 1  &  1 & 0  & 0  \\
    Level &  5  & 5  & 5  &  5 & 4  & 4  & 4  & 3  & 2  & 2  \\
    Extreme  &  -  & -  &  - & -  & $v_{37},v_{32} $ & $v_{38},v_{39} $ & - & -  &  - & -   \\
    \hline
    & $v_{11}$ & $v_{12}$ & $v_{13}$ & $v_{14}$ & $v_{15}$ & $v_{16}$ & $v_{17}$ & $v_{18}$ & $v_{19}$ & $v_{20}$\\
     \hline
    Rank &   0 & 1  & 1  & 3  &  1 & 1  &  1 & 1  & 1  & 1  \\
    Level &   1 &  4 & 3  & 3  & 2  & 2  &  3 & 2  &  3 & 1  \\
    Extreme  &  $v_{30},v_{43}$  & -  & $v_{31},v_{37}$  & $v_{38},v_{43}$  & $v_{38},v_{31}$  & $v_{30},v_{36}$  & $v_{29},v_{39}$  & $v_{29},v_{43}$  & $v_{30},v_{37}$  & $v_{29},v_{31}$  \\
    \hline
    & $v_{21}$ & $v_{22}$ & $v_{23}$ & $v_{24}$ & $v_{25}$ & $v_{26}$ & $v_{27}$ & $v_{28}$ & $v_{29}$ & $v_{30}$\\
     \hline
    Rank &  0  &  0 &  0 & 1  & 0  & 0  & 0  & 0  & 0  & 0  \\
    Level &  1  & 0  &  0 & 3  &  2 &  0 & 1  & 1  & 0  & 1  \\
    Extreme  &  $v_{30},v_{38}$  & $v_{30},v_{29}$  & $v_{11},v_{31}$  & $v_{11},v_{43}$  & $v_{11},v_{39}$  & $v_{17},v_{38}$  & $v_{17},v_{37}$  & $v_{17},v_{36}$  & $v_{17},v_{35}$  & $v_{11},v_{38}$  \\
    \hline
    & $v_{31}$ & $v_{32}$ & $v_{33}$ & $v_{34}$ & $v_{35}$ & $v_{36}$ & $v_{37}$ & $v_{38}$ & $v_{39}$ & $v_{40}$\\
     \hline
    Rank &  0  & 0  & 0  &  0 & 0  & 0  &  1 &1   &  1 & 0  \\
    Level &  0  &  1 &  0 &  2 &  0 & 1  & 2  & 3  & 1  &  1 \\
    Extreme  &  $v_{38},v_{13}$  & $v_{18},v_{38}$  & $v_{17},v_{38}$  & -  & $v_{30},v_{11}$  & $v_{30},v_{16}$  & $v_{30},v_{13}$  & $v_{30},v_{15}$  &  $v_{17},v_{11}$ &  - \\
    \hline
    & $v_{41}$ & $v_{42}$ & $v_{43}$ & $v_{44}$ & $v_{45}$ & $v_{46}$ & $v_{47}$ & $v_{48}$ & $v_{49}$ & $v_{50}$\\
     \hline
    Rank &  0  & 1  &  1 & 1  &  1 &  1 &  1 &  0 & 1  &  0 \\
    Level &  1  & 4  &  2 &  3 & 3  & 4  & 3  &  2 & 2  &  2 \\
    Extreme  &  -  & -  & $v_{18},v_{11}$  & -  & -  &  - & -  & -  &  - & -  \\
   % \hline
  \end{tabular}~\label{our-table}
  \caption{Information prepared by the ranking strategy. }
\label{our-table}
\end{table}

Now that the ranking strategy is completed, we will see how to use it in the main algorithm.

\section{Main algorithm}
The main algorithm is composed of a ranking strategy as a pre-processing stage, and a processing stage.

Once ranking strategy is over, in the main algorithm to compute $h_\triangle$ inscribed triangular separators we proceed as follows. We take a vertex $v \in \mathcal A$ (vertices other than Type-III vertices), and check $rank(v)$. According to Observation~\ref{empty-tri}, if $rank(v)<2$ then $v$ is a candidate location for a vertex of an inscribed triangular separator, and we continue by processing $v$. Otherwise, if $rank(v) \geq 2$ there is no need to process $v$. So assume that $rank(v)<2$. There are two lines passing through $v$. See Figure~\ref{walk-fig}. Let $\ell$ and $\ell'$ be these lines. Now we jump to $extreme(v)$ on $\ell$. (We do this procedure for $extreme(v)$ on $\ell'$ as well.) Numbering vertices on $\ell$ top down, let $extreme(v)$ be the vertex $u_j$ on $\ell$. Note that $rank(u_j)<2$ by the definition of the extreme label. Now we follow $u_j$ to get to the tangent point from $u_j$ to $CH_B$ on $\ell'$. Numbering vertices on $\ell'$ top down, let $w_i$ be this tangent point on $\ell'$. Using the structure of arrangement $\mathcal A$ we can reach $w_i$ in $O(1)$ time. According to Observation~\ref{empty-tri}, if $rank(w_i)<2$ then $w_i$ is a candidate location for a vertex of the separator, and therefore the triangle
$\overset{\triangle}{w_i v u_j}$ is an inscribed triangular separator.

\begin{figure}[h]
\centering
\includegraphics[width=7cm]{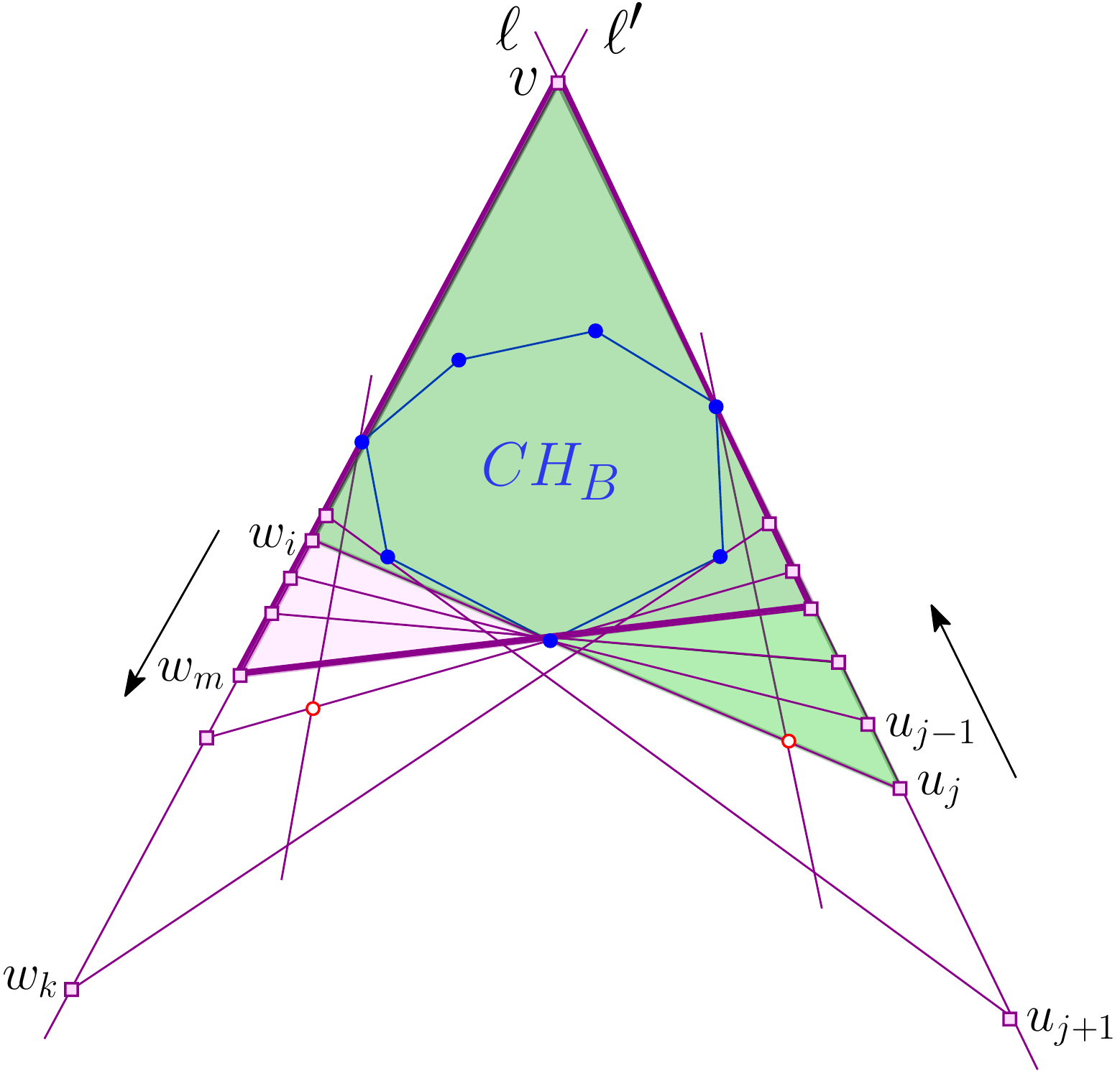}
\caption{Walking on $\mathcal A$ to report the inscribed triangular separators.}
\label{walk-fig}
\end{figure}

To find the rest of inscribed triangular separators with a vertex on $v$, we walk on $\ell'$ from $w_i$ downwards until we reach a vertex whose rank is greater than or equal to 2. Let $w_k$ be this vertex. We get the following result.

\begin{lemma}
The triangle $\overset{\triangle}{w_m v u(w_m)}$ is an inscribed triangular separator, where $i \leq m < k$ and $u(w_m)$ is the tangent point from $w_m$ on the line $\ell$.
\end{lemma}

\begin{proof}
The triangle $\overset{\triangle}{w_m v {u(w_m)}}$ can be partitioned into 4 pieces: $\triangle_v$ (the semi triangle based on $v$), $\triangle_{w_m}$, $\triangle_{u(w_m)}$, and $CH_B$. It is clear that $\overset{\triangle}{w_m v u(w_m)}$ contains all blue points. Hence, to show that $\overset{\triangle}{w_m v u(w_m)}$ is a separator it is sufficient to show that it contains no red points inside. Or let's say each partitioning part of it contains no red points. Since we are investigating the case that $CH_B$ does not contain any red points (otherwise, the decision version of triangular separability fails~\cite{conv-sep}), we focus on the other 3 parts. Note that $rank(v)<2$ and $rank(w_m)<2$ by the structure. Hence, $\triangle_v$ and $\triangle_{w_m}$ contains no red points. Further, we know that $rank(u_j)<2$, and therefore $\triangle_{u_j}$ does not contain any red points. It is easy to see that $\triangle_{u(w_m)}$ is a sub-region of $\triangle_{u_j}$. Thereby, it contains no red points. Now putting the non-red regions $CH_B, \triangle_v, \triangle_{w_m},$ and $\triangle_{u(w_m)}$ together, we get the triangular separator $\overset{\triangle}{w_m v u(w_m)}$.
\end{proof}

Thus, we can compute inscribed triangular separators with a vertex on $v$, each one in $O(1)$ time by using the information we have gathered in the ranking strategy. Assuming $h_\triangle$ to be the number of inscribed triangular separators, we can report them totally in $O(r \cdot (r+c_B+k)+h_\triangle)$ time. Assuming $r << b$ which is the case that often happens in practice, our proposed algorithm is a fast output-sensitive procedure to report the inscribed triangular separators.

\begin{theorem}
Given a set $R$ of $r$ red points and a set $B$ of $b$ blue points in a polygonal environment $\mathcal P$, we can compute the inscribed triangular separators in $O(r \cdot (r+c_B+k)+h_\triangle)$ time, where $c_B$ is the size of convex hull of $B$ and $h_\triangle$ is the number of these separators.
\end{theorem}

\section{Maximum triangular separator}
So far we have focused on finding triangular separators that separate all points in $B$ from $R$. Although this is the case that is often the desirable goal, it may not always possible due to distribution of points. For example, red points can tightly enclose $CH_B$ so that no triangular separator exists. In this section we focus on this situation, and we  adjust the goal, this time to find the triangle that separates most of the blue points from the red point, and it falls in a convex polygonal environment $\mathcal P$. We call this separator \emph{maximum triangular separator}, and design a constant-factor approximation algorithm to find it.

So assume that red points  tightly cover $CH_B$. However, no red point is inside it. By using the algorithm in~\cite{conv-sep} we can compute the convex polygon with the minimum number of edges separating  all the blue points from the red points. Let $\mathcal C$ denote this polygon. Note that $\mathcal C$ always exists since $CH_B$ has no red points inside, and $\mathcal C$ is $CH_B$ in the worst case. Now we enlarge $\mathcal C$ as follows. For each edge $e$ of $\mathcal C$, we move $e$ outwards in the orthogonal direction of its orientation until it hits a red point or the boundary of $\mathcal P$. Without loss of generality, from now we use $\mathcal C$ to refer to this enlarged convex separator. We have the following result.

\begin{lemma}~\label{le-1}
Let $\nabla_\mathrm{vrtx}$ be a maximum triangular separator with vertices chosen from vertices of $\mathcal C$, and let $\nabla_{\mathrm{OPT}}$ be a maximum triangular separator. Then, $|\nabla_\mathrm{vrtx}| \leq |\nabla_{\mathrm{OPT}}| \leq 7 |\nabla_\mathrm{vrtx}|$, where $|\mathcal S|$ denotes the number of blue points in $\mathcal S$.
\end{lemma}

\begin{proof}
Trivially $|\nabla_\mathrm{vrtx}| \leq |\nabla_{\mathrm{OPT}}|$.
To prove the second inequality, consider the intersections of $\nabla_{\mathrm{OPT}}$ and $\mathcal C$. The optimal separator  $\nabla_{\mathrm{OPT}}$ intersects $\mathcal C$ in at most six points.
Since $\mathcal C$ contains all blue points, and has a red point on each side (or touching the boundary of $\mathcal P$), then intersection of $\nabla_{\mathrm{OPT}}$ and $\mathcal C$ is a convex polygon with at most nine vertices. See~Figure~\ref{max-2}.
 This polygon can be partitioned into $7$ triangles at most with vertices chosen from the vertices of $\mathcal C$. Since the number of blue points in each triangle is less than or equal $|\nabla_\mathrm{vrtx}|$, the proof is completed.
\end{proof}

\begin{figure}[h]
\centering
\includegraphics[width=4cm]{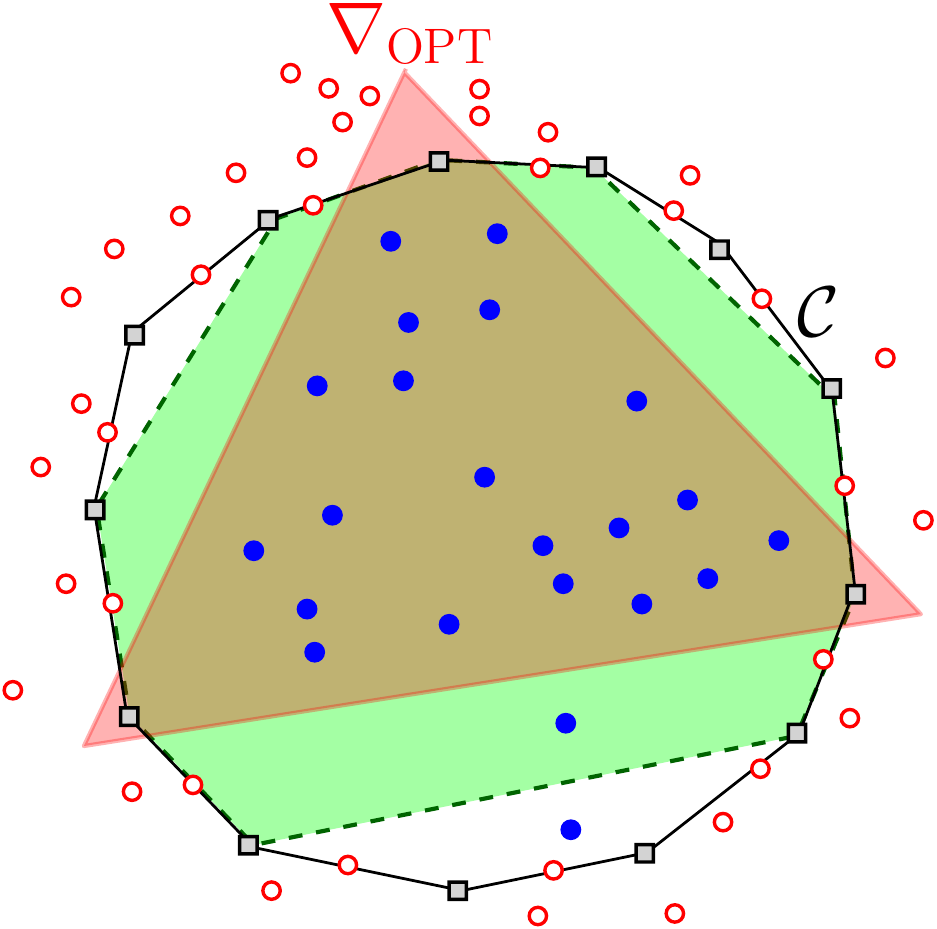}
\caption{Illustrating the intersection of $\nabla_{\mathrm{OPT}}$ and $\mathcal C$.}
\label{max-2}
\end{figure}

Let's take a step further and see what can be achieved.

\begin{lemma}~\label{le-2}
Let $u$ be a vertex on $\mathcal C$, and let $\nabla_\mathrm{u}$ be the maximum triangular separator whose vertices necessarily lie on $u$ and two other vertices of $\mathcal C$.  Then, $|\nabla_\mathrm{u}| \leq 2 |\nabla_\mathrm{vrtx}|$.
\end{lemma}

\begin{proof}
For a fixed vertex $u$, let $a, b$, and $c$ be the vertices of $\nabla_\mathrm{vrtx}$. Further, among the polygonal chains $\overarc{ab}$ (the chain between $a$ and $b$), $\overarc{bc}$, and $\overarc{ca}$, assume that $u$ lies on $\overarc{ab}$. By connecting $u$ to $a, b$ and $c$, the resulting triangles $\overset{\triangle}{a u c}$ and $\overset{\triangle}{u b c}$ cover all the blue points in $\nabla_\mathrm{vrtx}$. Thus, we have
$|\nabla_\mathrm{u}|   \leq  |\nabla_\mathrm{vrtx}| \leq |\overset{\triangle}{a u c}| + |\overset{\triangle}{u b c}| \leq 2 |\nabla_\mathrm{u}|$.
\end{proof}

Lemmas~\ref{le-1}~and~\ref{le-2} lead to the following corollary.

\begin{corollary}~\label{coro-main}
For a fixed vertex $u$, then $|\nabla_\mathrm{u}| \leq 14 |\nabla_{\mathrm{OPT}}|$.
\end{corollary}

So what is remained is to show how to compute $\nabla_\mathrm{u}$ for a fixed vertex $u$. In the following we propose a 2-approximation algorithm to do so.

\subsection{Approximation algorithm to compute $\nabla_\mathrm{u}$}
Let $n=r+b$.
For a fixed vertex $u$, there are $O(n^2)$ triangles with one vertex on $u$ and the other two vertices on the boundary of $\mathcal C$. Hence, by computing the number of blue points in each triangle in $O(f(n))$, the exact value of $|\nabla_\mathrm{u}|$ can be computed in $O(n^2\cdot f(n))$ time. However, in this part we focus to design a faster approximation algorithm to compute $|\nabla_\mathrm{u}|$.

Let $\overset{\triangle}{u v w}$ be a triangle with vertices on $u, v$, and $w$. Now we define $\mathcal T$ as the set of triangles $\overset{\triangle}{u a b}$, where there are $2^i$ vertices (for $0 \leq i \leq \log n$) on the chain $\overarc{ab}$. There are $O(n\log n)$ triangles in $\mathcal T$. Let $\nabla_{\mathcal T}$ be the triangle with the maximum number of blue points in $\mathcal T$.

\begin{lemma}
$|\nabla_{\mathcal T}| \leq |\nabla_\mathrm{u}| \leq 2|\nabla_{\mathcal T}|$.
\end{lemma}

\begin{proof}
Trivially $|\nabla_{\mathcal T}| \leq |\nabla_\mathrm{u}|$. To prove the second inequality, let $\nabla_\mathrm{u}$ be the triangle $\overset{\triangle}{u v w}$, and let $dist(v, w)$ be the number of vertices on the polygonal chain $\overarc{vw}$ on $\mathcal C$. Clearly, there exists some $j$, $0 \leq j \leq \log n$, such that $2^j \leq dist(v,w) \leq 2^{j+1}$. Now starting from $v$, walking in counterclockwise direction on $\mathcal C$, let $z$ be the $2^j$-th vertex that we visit. Further, starting from $w$, walking in clockwise direction on $\mathcal C$, let $z'$ be the $2^{j}$-th vertex that we visit. See Figure~\ref{app-prf}. Clearly the triangles $\overset{\triangle}{u v z}$ and $\overset{\triangle}{uwz'}$ cover $\nabla_\mathrm{u}$. Thus, we have
 $|\nabla_\mathrm{u}| \leq |\overset{\triangle}{uvz}|+|\overset{\triangle}{uwz'}| \leq 2 |\nabla_{\mathcal T}|$.
\end{proof}

\begin{figure}[h]
\centering
\includegraphics[width=5cm]{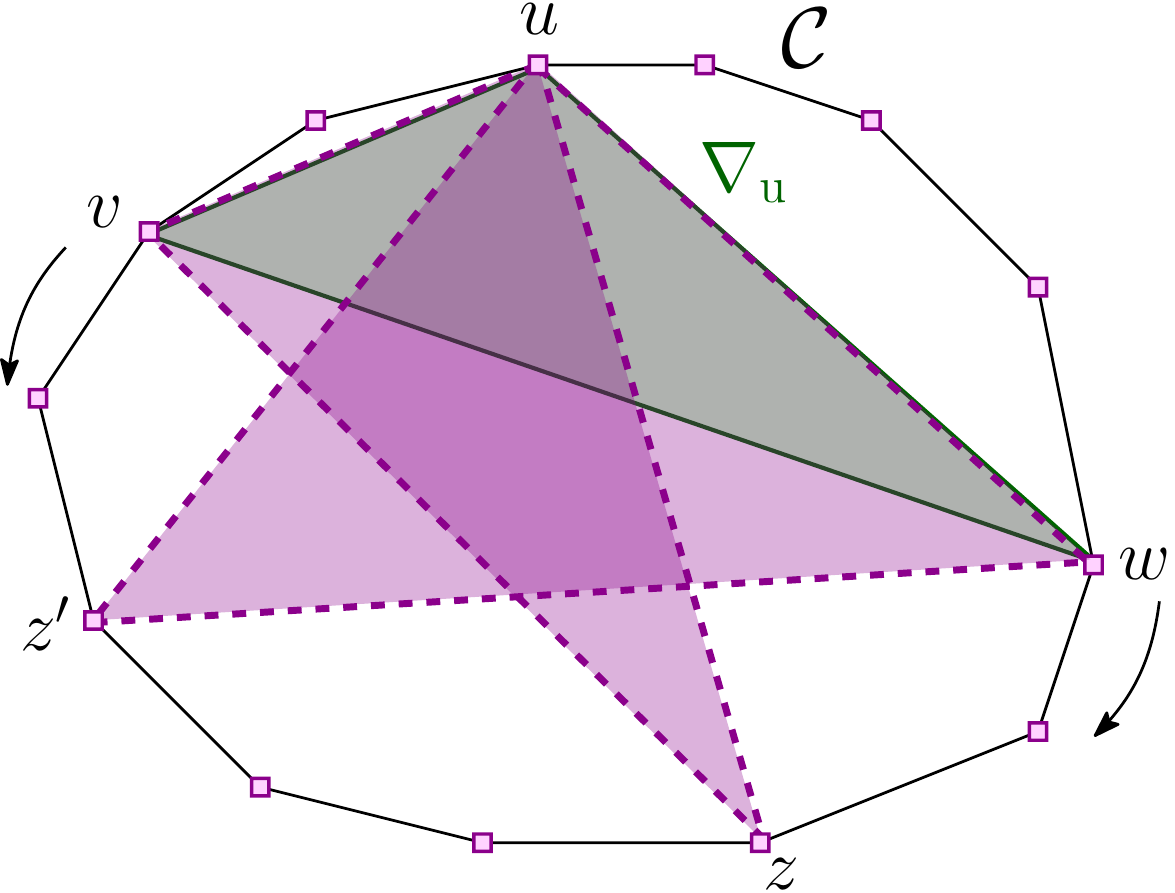}
\caption{Details in proving the approximation ratio.}
\label{app-prf}
\end{figure}

Now to complete our 2-approximation algorithm for finding $\nabla_\mathrm{u}$, we count the number of blue points within each triangle in $\mathcal T$, and take the triangle with the maximum value. It is known that

\begin{theorem}
\cite{cg-book}
\label{tc}
Given a set of $n$ points in the plane. Then for an integer $m$, $n\leq m\leq n^2$, we can pre-process the points in $O(m\log m)$ time and $O(m)$ space, so that counting the number of points within a query triangle can be done in
  $O(\dfrac{n}{\sqrt{m}}\log^3 \dfrac{m}{n})$ time.
\end{theorem}

Thus, by using~Theorem~\ref{tc} we spend an $O(m\log m)$ pre-processing time. This way, for each triangle in $\mathcal T$ we can count the number of blue points inside in $O(\dfrac{n}{\sqrt{m}}\log^3 \dfrac{m}{n})$ time. Since there are $O(n \log n)$ triangles in $\mathcal T$, the overall time complexity of the algorithm is
$O(m\log m + (\dfrac{n}{\sqrt{m}}\log^3 \dfrac{m}{n}) \cdot (n \log n))$. There is a tradeoff between $m$ and $n$. By choosing $m = n^{4/3} \log^2 n$, we can achieve the overall time complexity of $O(n^{4/3} \log^3 n)$ for a 2-approximation algorithm to compute $\nabla_\mathrm{u}$. Further, according to Corollary~\ref{coro-main}, $\nabla_\mathrm{u}$ is a 14-approximation of the maximum triangular separator. This results in a 28-approximation algorithm to compute the maximum triangular separator. The algorithm also handle the case (without altering the time complexity) where the sets $R$ and $B$ are distributed in a convex polygonal environment $\mathcal P$, leading to the following theorem.

\begin{theorem}
Let $B$ be a set of $b$ blue points and  $R$ be a set of $r$ red blue points in the plane (or a convex polygonal environment) with $n = r+b$. Then when the red points  tightly enclose $CH_B$ but do not exist inside, there is a 28-approximation algorithm to compute the maximum triangular separator in $O(n^{4/3} \log^3 n)$ time.
\end{theorem}

\section{Concluding remarks}
In this paper we have studied triangular separability of bichromatic point sets $B$ and $R$. Modeling the geographical constraint of the environment, first we assumed that points are distributed in a polygonal environment $\mathcal P$, and proposed an output-sensitive algorithm to find the inscribed triangular separators. Our algorithm can also handle the case where $\mathcal P$ is a simple closed curvature defined by a collection of circular arcs. When red points tightly enclose convex hull of $B$ although no red points are inside the hull, then there does not exist any inscribed triangular separators. We have proposed a constant-factor approximation algorithm to find the maximum triangular separator in this case. Computing the maximum triangular separator has a more difficult nature than the previous problem. We have designed a set of $O(n \log n)$ candidate triangular separators, and use a range counting algorithm to find the number of blue points inside these query triangles. The query time is  $O(n^{1/3} poly\log n)$. It has been shown that it  is impossible to achieve $O(poly\log n)$ time for triangular counting queries using $O(n poly\log n)$ space, even when the query triangle contains the origin~\cite{tri-query,tri-thesis}. This gives a brighter insight into the difficult nature of the problem. Therefore, achieving an $O(n poly\log n)$ time approximation algorithm to compute the maximum triangular separator is an interesting open problem.

\bibliographystyle{plain}
\bibliography{tri-ref}{}

\begin{thebibliography}{10}

\bibitem{tri-ag}
A.~Aggarwal and J.~Park.
\newblock Notes on searching in multidimensional monotone arrays.
\newblock In {\em 29th Annual Symposium on Foundations of Computer Science
  (FOCS)}, pages 497--512, 1988.

\bibitem{L-cover}
S.~W. Bae, C.~Lee, H.-K. Ahn, S.~Choi, and K.-Y. Chwa.
\newblock Maintaining extremal points and its applications to deciding optimal
  orientations.
\newblock In {\em 18th International Symposium on Algorithms and Computation
  (ISAAC)}, pages 788--799, 2007.

\bibitem{tri-query}
N.~M. Benbernou, M.~Ishaque, and D.~L. Souvaine.
\newblock Data structures for restricted triangular range searching.
\newblock In {\em 20th Canadian Conference on Computational Geometry (CCCG)},
  pages 15--18, 2008.

\bibitem{bose2}
P.~Bose and J.-L.~De Carufel.
\newblock Minimum-area enclosing triangle with a fixed angle.
\newblock {\em Computational Geometry}, 47(1):90 -- 109, 2014.

\bibitem{bose1}
P.~Bose, M.~Mora, C.~Seara, and S.~Sethia.
\newblock On computing enclosing isosceles triangles and related problems.
\newblock {\em International Journal of Computational Geometry \&
  Applications}, 21(1):25 -- 45, 2011.

\bibitem{tri-chang}
J.~S. Chang and C.~K. Yap.
\newblock A polynomial solution for potato peeling and other polygon inclusion
  and enclosure problems.
\newblock In {\em 25th Annual Symposium on Foundations of Computer Science
  (FOCS)}, pages 408--416, 1984.

\bibitem{demaine-chord}
E.~D. Demaine, J.~Erickson, F.~Hurtado, J.~Iacono, S.~Langerman, H.~Meijer,
  M.~Overmars, and S.~Whitesides.
\newblock Separating point sets in polygonal environments.
\newblock {\em International Journal of Computational Geometry \&
  Applications}, 15(4):403--419, 2005.

\bibitem{conv-sep}
H.~Edelsbrunner and F.~P. Preparata.
\newblock Minimum polygonal separation.
\newblock {\em Information and Computation}, 77(3):218--232, 1988.

\bibitem{simple-fekete}
S.~Fekete.
\newblock On the complexity of min-link red-blue separation. {M}anuscript.
\newblock 1992.

\bibitem{tri-thesis}
M.~Ishaque.
\newblock {\em Geometric Data Structures}.
\newblock PhD Thesis, Tufts University, 2010.

\bibitem{Megiddo-LP}
N.~Megiddo.
\newblock Linear-time algorithms for linear programming in $\mathbb{R}^3$ and
  related problems.
\newblock {\em SIAM Journal on Computing}, 12(4):759--776, 1983.

\bibitem{simple-apprx-mitchel}
J.~S.~B. Mitchell.
\newblock Approximation algorithms for geometric separation problems.
\newblock Technical report, State University of New York at Stony Brook, 1993.

\bibitem{tri-rourke}
J.~O'Rourke, A.~Aggarwal, S.~Maddila, and M.~Baldwin.
\newblock An optimal algorithm for finding minimal enclosing triangles.
\newblock {\em Journal of Algorithms}, 7(2):258 -- 269, 1986.

\bibitem{sep-cir}
J.~O'Rourke, S.~Rao~Kosaraju, and N.~Megiddo.
\newblock Computing circular separability.
\newblock {\em Discrete Computational Geometry}, 1(1):105--113, 1986.

\bibitem{cg-book}
J.-R. Sack and J.~Urrutia, editors.
\newblock {\em Handbook of Computational Geometry}.
\newblock North-Holland Publishing Co., Amsterdam, The Netherlands, 2000.

\bibitem{farnaz-rect}
F.~Sheikhi and A.~Mohades.
\newblock Planar maximum-box problem revisited.
\newblock {\em Theoretical Computer Science}, 729:57 -- 67, 2018.

\bibitem{sep-ours}
F.~Sheikhi, A.~Mohades, M.~de~Berg, and M.~Davoodi.
\newblock Separating bichromatic point sets by {L}-shapes.
\newblock {\em Computational Geometry}, 48(9):673 -- 687, 2015.

\bibitem{rotating-cal}
G.~Toussaint.
\newblock Solving geometric problems with the rotating calipers.
\newblock In {\em IEEE MELECON}, pages A10.02/1--4, 1983.

\bibitem{sep-rect}
M.~van Kreveld, T.~van Lankveld, and R.~Veltkamp.
\newblock Identifying well-covered minimal bounding rectangles in 2{D} point
  data.
\newblock In {\em 25th European Workshop on Computational Geometry (EuroCG)},
  pages 277--280, 2009.

\bibitem{kreveld-tech}
T.~van Lankveld, M.~van Kreveld, and R.~C. Veltkamp.
\newblock Identifying rectangles in laser range data for urban scene
  reconstruction.
\newblock Technical report, Utrecht University, Utrecht, The Netherlands, 2011.

\end{thebibliography}
  \end{document}